\documentclass[10pt,conference]{IEEEtran}

\usepackage[utf8]{inputenc}
\usepackage{amsmath}
\usepackage{amssymb}
\usepackage{amsthm}
\usepackage{pgfplots}

\newcommand{\nActive}{k}
\newcommand{\nInactive}{N}
\newcommand{\potActive}[1]{N_{#1}}
\newcommand{\setPotActive}[1]{\mathcal{N}_{#1}}
\newcommand{\setPotActiveChosen}[1]{\mathcal{U}_{#1}}

\newcommand{\potActiveNotChosen}[1]{B_{#1}}
\newcommand{\nSlots}{\ell}
\newcommand{\indexSlots}{i}
\newcommand{\shareInactive}{C}
\newcommand{\errorprob}{\varepsilon}
\newcommand{\chooseProbability}{p}
\newcommand{\chooseProbabilityInv}{q}
\newcommand{\Probability}{\text{P}}
\newcommand{\Expectation}{\text{E}}
\newcommand{\chooseActiveIndicator}[1]{A_{#1}}
\newcommand{\shortTermVarOne}{n}
\newcommand{\shortTermVarTwo}{x}
\newcommand{\shortTermRV}{X}
\newcommand{\shortTermVarSubgauss}{\alpha}
\newcommand{\channelIn}[1]{X_{#1}}
\newcommand{\channelOut}{Y}
\newcommand{\channelOutIntermediate}{\tilde{Y}}
\newcommand{\channelNoise}{Z}
\newcommand{\channelPowerLim}{P}
\newcommand{\channelrep}{m}
\newcommand{\channelErrorprob}{\delta}
\newcommand{\channelInstant}{t}
\newcommand{\subgaussnormMax}{K}
\newcommand{\hoeffdingTypeConst}{c}
\newcommand{\channelBlockLength}{n}
\newcommand{\channelDecoder}{D}
\newcommand{\channelEncoder}[1]{E_{#1}}
\newcommand{\channelInputAlphabet}[1]{\mathcal{X}_{#1}}
\newcommand{\channelOutputAlphabet}{\mathcal{Y}}
\newcommand{\nChannelEndpoints}{R}
\newcommand{\indexChannelEndpoints}{r}
\newcommand{\channel}{\mathcal{C}}
\newcommand{\channelMessage}[1]{x_{#1}}
\newcommand{\channelMessageDecoded}{y}
\newcommand{\channelKernel}[1]{W_{#1}}
\newcommand{\disjunctionFunction}{\text{dis}}

\newcommand{\eulerNum}{e}

\newcommand{\absolute}[1]{\left|#1\right|}
\newcommand{\cardinality}[1]{\left|#1\right|}
\newcommand{\naturals}{\mathbb{N}}
\newcommand{\reals}{\mathbb{R}}
\newcommand{\subgaussnorm}[1]{\left\|#1\right\|_{\Psi_2}}

\newcommand{\Markov}{Markov}

\newtheorem{theorem}{Theorem}

\newtheorem{lemma}{Lemma}
\newtheorem{cor}{Corollary}
\newtheorem{Prob}{Problem}
\newtheorem{Assume}{Assumption}
\newtheorem{remark}{Remark}
\newtheorem{definition}{Definition}

\title{User Activity Detection via Group Testing and Coded Computation}
\author{Matthias Frey, Igor Bjelaković and Sławomir Stańczak\\Technische Universität Berlin}

\begin{document}
\maketitle
\begin{abstract}
Inspired by group testing algorithms and the coded computation paradigm, we propose and analyze a novel multiple access scheme for detecting active users in large-scale networks. The scheme consists of a simple randomized detection algorithm that uses computation coding as intermediate steps for computing logical disjunction functions over the multiple access channel (MAC). First we show that given an efficient MAC code for disjunction computation the algorithm requires $O(k \log (\frac{N}{k }))$ decision steps for detecting $k$ active users out of $N+k$ users. Subsequently we present a simple suboptimal code for a class of MACs with arbitrarily varying sub-gaussian noise that uniformly requires $O (k \log (N) \max \{ \log k , \log \log N \} )$ channel uses for solving the activity detection problem. This shows that even in the presence of noise an efficient detection of active users is possible. Our approach reveals that the true crux of the matter lies in constructing efficient codes for computing disjunctions over a MAC.
\end{abstract}
\section{Introduction}
\label{sec:intro}

\subsection{Motivation}
Novel Machine Type Communication (MTC) applications, which often involve a massive deployment of MTC devices, pose some fundamental challenges. One of these is the development of new uncoordinated random access schemes that incur a small computation and communication overhead. An inherent feature of any random access scheme is the ability to identify the set of active devices that compete for access to the channel. This problem is referred to as the user activity detection problem. Unfortunately, conventional pilot-based approaches to this problem entail unacceptable signaling overhead in envisioned massive MTC scenarios, because the overhead of such schemes scales linearly with the total number of deployed devices. On the positive side, in most MTC applications, the set of active devices at any given point in time will be smaller by orders of magnitude than the set of all deployed devices. This can be exploited to reduce the overhead significantly so that it scales logarithmically with the total number of deployed devices.

\subsection{Related Work}
The problem of user activity detection in large-scale networks has been already addressed in the context of compressive sensing in \cite{Kueng_Jung_16,Kueng_Jung_16_long} where the authors show the existence of a non-adaptive optimal activity detection strategy. The approach in \cite{Kueng_Jung_16,Kueng_Jung_16_long} relies heavily on methods from high dimensional probability and on methods for estimating certain partial sums affected by noise.

In contrast to this, we propose to decompose the user activity detection problem into two sub-problems: 1) The problem of finding a suitable computational channel code, and 2) the problem of using this code to detect the set of active nodes in a manner that is efficient in terms of communication overhead and computational resources used.

Computation over multiple-access channels, to the best of our knowledge, was first proposed for a special case in~\cite{Korner_Marton_79} and later formalized for a more general setting in~\cite{Nazer_Gastpar_07}. A special case that has received much attention is the computation of the modulo sum of sources over multiple-access channels because it yields a more efficient approach to network coding which is called Physical Layer Network Coding~\cite{Nazer_Gastpar_11}. For the binary case, this can be viewed as computation coding for the logical Exclusive Or function. We, in contrast, need a computational channel code which computes logical disjunctions (or equivalently conjunctions).

The problem of determining the set of active nodes from such disjunctions is remarkably similar to Combinatorial Group Testing problems which were first proposed as a way to conduct syphilis tests for draftees of the U.S. armed forces during World War II more efficiently by pooling their blood samples in clever ways~\cite{Du_Hwang_99}. Although this idea was not implemented in the end, the approach remains relevant in biology and bioinformatics e.g. for drug screening or DNA sequencing~\cite{Kainkaryam_Woolf_09,Hwang_Liu_04}. Existing deterministic and randomized approaches in Combinatorial Group Testing~\cite{Du_Hwang_99}, however, usually aim at constructing pooling strategies that work for all distributions of a given number of active devices (or, using the terminology from the biology applications, defective samples) as opposed to constructing a randomized strategy which works well with high probability for an unknown, but fixed set of active devices. There is an approach to pooling samples for DNA tests called Random Incidence Designs~\cite{Hwang_Liu_04} (a variation of which appeared first in a purely mathematical context in~\cite{Erdos_Renyi_63}) that is similar to what we propose; however, the authors of \cite{Hwang_Liu_04} do not give a comparable analysis aimed at detecting a number of active devices (or in their language: positive clones) which is much smaller than the number of total devices. We point out that our results in Theorem~\ref{main-theorem} and Corollary~\ref{exact-set-corollary} may also be of interest as schemes for Combinatorial Group Testing.

A result similar to our Corollary~\ref{exact-set-corollary} appeared in an independent work~\cite{Zhang_Huang_17} using a slightly different scheme; we do point out, however, that our result is more general in the sense that we do not restrict the growth of the number of active devices depending on the number of total devices, while the result in~\cite{Zhang_Huang_17}, on the other hand, offers a slightly better constant ($(\log 2)^{-2} \approx 2.1$ where we have $\eulerNum \approx 2.7$).

\subsection{Outline}
In the following section, we state the underlying problem, while Section~\ref{scheme} introduces the framework for activity detection. In Section~\ref{sec:analysis}, we provide an analysis of the proposed scheme. Section~\ref{sec:simulations} contains simulation results evaluating the impact of true common randomness versus using pseudo-random numbers for the execution of the scheme. In Section~\ref{coding}, we give a simple suboptimal code enabling detection of active users simultaneously and uniformly for a class of MACs with additive sub-gaussian noise satisfying an upper bound on the sub-gaussian norm of the noise random variables.

\subsection{Notation}
We use $\eulerNum$ to denote Euler's number and $\log$ to denote the logarithm with base $\eulerNum$.

\section{Problem Statement and Objectives}
\label{sec:prob}
We consider a multiple access channel with a single receiver and $\nInactive + \nActive$ transmitters (also called nodes). The nodes are indexed by numbers $1, \dots, \nActive + \nInactive$ in some arbitrary order, which is known both to the receiver and to the nodes. Of these $\nInactive + \nActive$ nodes, $\nInactive$ are considered to be \emph{inactive}, while the remaining $\nActive$ are \emph{active}. Each node knows only its own state which is either active or inactive. The notion of being active could mean that a node wishes to make a subsequent transmission; but for the sake of this problem, being active or inactive is just a labeling dividing the nodes into two classes. The problem of interest in this paper is the following:
\begin{Prob}
\label{prob:general}
Given a total number $\channelBlockLength\in\naturals$ of channel uses, find a scheme that enables the nodes to encode the information whether or not they are active as $\channelBlockLength$ channel inputs in such a way that the receiver can determine with a sufficiently small probability of error which of the $\nActive + \nInactive$ nodes are active.
\end{Prob}
Our objective is therefore to design a multiple access communication scheme that enables the receiver to recover the indices of the active nodes from $\channelBlockLength$ uses of the multiple access channel, while meeting a given requirement on the error probability (for the exact definition of the error probability, we refer the reader to Theorem \ref{main-theorem}). At this point, it is important to emphasize that our goal is \emph{not} to design an optimal scheme for a given channel in the sense of minimizing the number of channel uses for some given error probability among all possible multiple access schemes. In fact, in this paper, we introduce a computation code as a layer of abstraction between our proposed detection scheme and the channel.
\begin{definition}
\label{dfn:disjunctioncode}
Let $\channel$ be a memoryless multiple-access channel defined by input alphabets $\channelInputAlphabet{1}, \dots, \channelInputAlphabet{\nChannelEndpoints}$, output alphabet $\channelOutputAlphabet$ (the alphabets need not be discrete) and stochastic kernels $(\channelKernel{\channelInstant})_{\channelInstant \in \{1,\dots\}}$ that define the transition from channel inputs $(\channelIn{1,\channelInstant},\dots,\channelIn{\nChannelEndpoints,\channelInstant}) \in \channelInputAlphabet{1} \times \dots \times \channelInputAlphabet{\nChannelEndpoints}$ to a channel output $\channelOut_{\channelInstant} \in \channelOutputAlphabet$ at time instant $\channelInstant$ (in the case of discrete input and output alphabets, the stochastic kernels can be thought of as matrices of conditional probabilities).

An \emph{$(\channelBlockLength,\nSlots,\channelErrorprob)$-disjunction code} for $\channel$ consists of encoding functions for $\indexChannelEndpoints \in \{1, \dots, \nChannelEndpoints\}$
\[
\channelEncoder{\indexChannelEndpoints}: \{\text{\emph{true}},\text{\emph{false}}\}^\nSlots \rightarrow \channelInputAlphabet{\indexChannelEndpoints}^\channelBlockLength
\]
and a decoding function
\[
\channelDecoder: \channelOutputAlphabet^\channelBlockLength \rightarrow \{\text{\emph{true}},\text{\emph{false}}\}^\nSlots,
\]
such that
\begin{align*}
\Probability \big(\channelDecoder(\channelOut_1, \dots, \channelOut_\channelBlockLength) \neq \disjunctionFunction(\channelMessage{1}, \dots, \channelMessage{\nChannelEndpoints}) \big| &\phantom{\big)}\\
\forall \indexChannelEndpoints (\channelIn{\indexChannelEndpoints,1}, \dots, \channelIn{\indexChannelEndpoints,\channelBlockLength}) = \channelEncoder{\indexChannelEndpoints}(\channelMessage{\indexChannelEndpoints}) &\big) \leq \channelErrorprob
\end{align*}
for every $\channelMessage{1}, \dots, \channelMessage{\nChannelEndpoints} \in \{\text{\emph{true}},\text{\emph{false}}\}^\nSlots$, where
\[
\disjunctionFunction: \left(
\begin{pmatrix}                           
\channelMessage{1}^{(1)} \\ \vdots \\ \channelMessage{1}^{(\nSlots)}
\end{pmatrix}
, \dots,
\begin{pmatrix}
\channelMessage{\nChannelEndpoints}^{(1)} \\ \vdots \\ \channelMessage{\nChannelEndpoints}^{(\nSlots)}
\end{pmatrix}
\right)
\mapsto
\begin{pmatrix}
\channelMessage{1}^{(1)} \vee \dots \vee \channelMessage{\nChannelEndpoints}^{(1)} \\
\vdots \\
\channelMessage{1}^{(\nSlots)} \vee \dots \vee \channelMessage{\nChannelEndpoints}^{(\nSlots)}
\end{pmatrix}
\]
calculates logical disjunctions component-wise.
\end{definition}
Throughout Sections~\ref{scheme} and~\ref{sec:analysis}, we make the following assumption:
\begin{Assume}
\label{as:channelcode_existence}
There exists a family of $(\channelBlockLength,\nSlots,\channelErrorprob)$-disjunction codes for arbitrarily small $\channelErrorprob > 0$ and sufficiently large $\nSlots \in \naturals$ such that $\channelBlockLength/\nSlots$ is bounded.
\end{Assume}
In Section~\ref{coding}, we present crude codes for a class of additive noise channels with unbounded $\channelBlockLength/\nSlots$ that do, however, achieve practically relevant bounds on the number of channel uses needed to guarantee low overall errors for our scheme. Classifying channels for which good disjunction codes in the sense of Assumption~\ref{as:channelcode_existence} exist and finding such codes remain open problems for future research.
\section{Scheme}
\label{scheme}
The scheme consists of $\nSlots$ slots in each of which one dis\-junc\-tion needs to be computed over the channel. We use an $(\channelBlockLength,\nSlots,\channelErrorprob)$-disjunction code as in Assumption~\ref{as:channelcode_existence} to convey the resulting $\nSlots$ disjunctions.
\begin{remark}
There will be no dependence between different slots, so a block coding scheme as in Definition~\ref{dfn:disjunctioncode} that calls for simultaneous encoding of the input bits and simultaneous decoding of the output bits can be used.
\end{remark}
The receiver maintains a set of \emph{potentially active} nodes. All active nodes are contained in this set at all times. By $\setPotActive{\indexSlots}$ we denote the set of potentially active nodes after slot $\indexSlots$. We define the number of potentially but not actually active nodes after slot $\indexSlots$ as $\potActive{\indexSlots} := \cardinality{\setPotActive{\indexSlots}}-\nActive$. The scheme works as follows:
\begin{enumerate}
 \item At the beginning, all $\nInactive + \nActive$ nodes are defined to be potentially active, i.e. $\setPotActive{0}$ consists of all $\nInactive+\nActive$ nodes. We start the scheme with slot $\indexSlots := 1$. 
 \item \label{scheme-slot-begin} At the beginning of slot $\indexSlots$, a set of \emph{chosen} nodes is determined randomly. We assume common randomness at the receiver and all the nodes, i.e. every node and the receiver know which nodes are chosen. Each node is chosen independently with probability $\chooseProbability$. We often use $\chooseProbabilityInv := 1 - \chooseProbability$ to denote the probability that a given node is not chosen. By $\setPotActiveChosen{\indexSlots}$, we denote the set of nodes which are chosen but not active.
 \item Each node uses the following rule to determine the truth value that it contributes to the logical disjunction conveyed over the channel in slot $\indexSlots$: 
\begin{itemize}
\item[a)] It sends the message \emph{true} if it is both active and chosen.
\item[b)] It sends the message \emph{false} otherwise.
\end{itemize}
 \item Upon decoding the logical disjunction of the transmitted messages, the receiver applies the following rule:
\begin{itemize}
\item[a)] If the receiver decodes \emph{true}, it can only conclude that at least one of the chosen nodes is active. It discards this information and sets $\setPotActive{\indexSlots} := \setPotActive{\indexSlots-1}$.
 \item[b)] If the receiver decodes \emph{false}, it concludes that none of the chosen nodes are active and removes all chosen nodes from the set of potentially active nodes by setting $\setPotActive{\indexSlots} := \setPotActive{\indexSlots-1} \setminus \setPotActiveChosen{\indexSlots}$. 
\end{itemize}
 \item If $\indexSlots=\nSlots$, the scheme terminates with the set $\setPotActive{\nSlots}$ as output. Otherwise, we increment $\indexSlots$ by $1$ and continue from step~\ref{scheme-slot-begin}.
\end{enumerate}
It is shown in Section~\ref{sec:analysis} that for any $\errorprob > 0$, $\shareInactive > 0$ and sufficiently large $\nSlots$, the cardinality of the output $\setPotActive{\nSlots}$ is less than $\nActive(\shareInactive+1)$ with probability $1-\errorprob$. Also, the output clearly contains all active nodes.
\begin{remark}
Since $\setPotActive{\nSlots}$ has a size linear in $\nActive$, narrowing down this set to the exact set of active nodes can be done in time linear in $\nActive$; one possibility would be to count the number of active nodes in certain subsets of $\setPotActive{\nSlots}$ or to at least obtain parity information about these numbers and then solve a linear equation system to determine which nodes are active. An additional coding technique would have to be used that can count or obtain this parity information in an efficient way, for which the schemes employed in Physical Layer Network Coding~\cite{Nazer_Gastpar_11} may be promising candidates. This issue is however not addressed in this paper. 
\end{remark}
\begin{remark}
A major advantage of this approach is that the algorithm used at the nodes and the receiver is simple in terms of both conceptual complexity and computational resources used. The amount of memory needed in total as well as the computational effort for the computation during each slot is constant at the transmitters and linear in $\nInactive + \nActive$ at the receiver.
\end{remark}
\section{Analysis}
\label{sec:analysis}
In this section, we analyze the performance of the scheme described in Section~\ref{scheme}. We use the definitions introduced before.
\begin{theorem}
\label{main-theorem}
There exists $\chooseProbability \in (0,1)$ such that $\Probability(\potActive{\nSlots} \geq \shareInactive \nActive) \leq \errorprob$ provided that
\begin{equation}
\label{main-theorem-hypothesis}
\nSlots \geq \eulerNum (\nActive+1) \left( \log\frac{\nInactive}{\nActive} + \log\frac{1}{\errorprob} + \log\frac{1}{\shareInactive} \right).
\end{equation}
\end{theorem}
By taking $\shareInactive = \nActive^{-1}$, we obtain the following corollary:
\begin{cor}
\label{exact-set-corollary}
There exists $\chooseProbability \in (0,1)$ such that $\Probability(\potActive{\nSlots} > 0) \leq \errorprob$ provided that
\begin{equation}
\label{exact-set-corollary-hypothesis}
\nSlots \geq \eulerNum (\nActive+1) \left( \log\nInactive + \log\frac{1}{\errorprob} \right).
\end{equation}
\end{cor}
Thus, for the cost of a slightly worse bound with regard to dependency on $\nActive$, we can determine the exact set of active nodes without having to employ another scheme to make a final determination.

In order to prove Theorem~\ref{main-theorem}, we will bound the expectation of $\potActive{\nSlots}$ in the following lemmas and then use \Markov's inequality to bound the probability.
\begin{lemma}
\label{lemma-expectation}
$\Expectation \potActive{\nSlots} \leq \shareInactive \nActive \errorprob$ if
\begin{equation}
\label{lemma-expectation-hypothesis}
\nSlots \geq \left(-\frac{1}{\log\left(1-\chooseProbability \chooseProbabilityInv^\nActive\right)} \right) \left( \log \frac{\nInactive}{\shareInactive \nActive \errorprob} \right).
\end{equation}
\end{lemma}
A similar result already appeared in the context of pooling DNA tests in~\cite[Corollary 3.4]{Hwang_Liu_04}, but for the sake of completeness, we give here an easier and more self-contained proof.
\begin{proof}
We define random variables $(\chooseActiveIndicator{\indexSlots})_{\indexSlots \in \{1, \dots, \nSlots\}}$ by
\[
\chooseActiveIndicator{\indexSlots} :=
\begin{cases}
 0, &\text{only inactive nodes are chosen in slot } \indexSlots \\
 1, &\text{otherwise.}
\end{cases}
\]
Note that $\Expectation(\chooseActiveIndicator{\indexSlots}) = 1 - \chooseProbabilityInv^\nActive$ and $\Expectation(1-\chooseActiveIndicator{\indexSlots}) = \chooseProbabilityInv^\nActive$.

Denote with $\potActiveNotChosen{\indexSlots}$ the number of potentially but not actually active nodes which are not chosen in slot $i$. Clearly, for any $\shortTermVarOne$ and $\indexSlots$ under the condition $\potActive{\indexSlots-1} = \shortTermVarOne$, $\potActiveNotChosen{\indexSlots}$ is binomially distributed with $\shortTermVarOne$ trials and success probability $\chooseProbabilityInv$. Clearly,
\[ \potActive{\indexSlots} = \chooseActiveIndicator{\indexSlots} \potActive{\indexSlots-1} + (1-\chooseActiveIndicator{\indexSlots}) \potActiveNotChosen{\indexSlots}. \]
We observe that by the law of total expectation for all $\indexSlots$
\begin{align*}\Expectation \potActiveNotChosen{\indexSlots} &=
\sum_\shortTermVarOne \Probability(\potActive{\indexSlots-1}=\shortTermVarOne) \Expectation(\potActiveNotChosen{\indexSlots} | \potActive{\indexSlots-1} = \shortTermVarOne) \\ &=
q \sum_\shortTermVarOne \shortTermVarOne \Probability(\potActive{\indexSlots-1}=\shortTermVarOne) =
q \Expectation \potActive{\indexSlots-1}
\end{align*}
and further that because nodes are chosen or not chosen independently in different slots, for all $\indexSlots$, $\chooseActiveIndicator{\indexSlots}$ is independent of $\potActive{\indexSlots'}$ for $\indexSlots' < \indexSlots$ and because different nodes are chosen or not chosen independently even within the same slot, $\chooseActiveIndicator{\indexSlots}$ is independent of $\potActiveNotChosen{\indexSlots'}$ for $\indexSlots' \leq \indexSlots$ (although, of course, for other values of $\indexSlots'$ there may be dependences).

We can now use these observations to calculate
\begin{align*}
\Expectation \potActive{\indexSlots}
&= \Expectation \chooseActiveIndicator{\indexSlots} \Expectation \potActive{\indexSlots - 1} + \Expectation (1-\chooseActiveIndicator{\indexSlots}) \Expectation \potActiveNotChosen{\indexSlots} \\
&= (1-\chooseProbabilityInv^\nActive) \Expectation \potActive{\indexSlots - 1} + \chooseProbabilityInv^\nActive \cdot \chooseProbabilityInv \Expectation \potActive{\indexSlots - 1} \\
&= \Expectation \potActive{\indexSlots - 1} (1-\chooseProbability \chooseProbabilityInv^\nActive).
\end{align*}
From this recursion and the definition $\potActive{0}=\nInactive$, we conclude
\[ \Expectation \potActive{\indexSlots} = \nInactive \left(1- \chooseProbability \chooseProbabilityInv^\nActive \right)^\indexSlots. \]
In particular, for $\indexSlots=\nSlots$, we have
\[ \Expectation \potActive{\nSlots} = \nInactive \left(1- \chooseProbability \chooseProbabilityInv^\nActive \right)^\nSlots \leq \shareInactive \nActive \errorprob. \]
Taking the logarithm on both sides and solving for $\nSlots$ yields~(\ref{lemma-expectation-hypothesis}), which completets the proof.
\end{proof}
Next we show how to choose $\chooseProbability$ so as to minimize the factor $(-\log(1-\chooseProbability \chooseProbabilityInv^\nActive))^{-1}$ in~(\ref{lemma-expectation-hypothesis}), which is equivalent to maximizing $\chooseProbability \chooseProbabilityInv^\nActive$. The derivative has zeros at $\chooseProbability=1$ and $\chooseProbability=(\nActive+1)^{-1}$. Since $\chooseProbability \chooseProbabilityInv^\nActive = 0$ for $\chooseProbability=0$ and $\chooseProbability=1$, the maximum must be at $\chooseProbability=(\nActive+1)^{-1}$.
\begin{lemma}
\label{lemma-limit}
If $\chooseProbability = (\nActive+1)^{-1}$, then
\[ -\frac{1}{\log\left( 1-\chooseProbability \chooseProbabilityInv^\nActive \right)} < \eulerNum (\nActive + 1). \]
\end{lemma}
\begin{proof}
Using the fact that $-\log(1-\shortTermVarTwo) \geq \shortTermVarTwo$ for $\shortTermVarTwo \leq 1$, and substituting $1-\chooseProbability$ for $\chooseProbabilityInv$ as well as $(\nActive+1)^{-1}$ for $\chooseProbability$, we obtain
\begin{align*}
-\frac{1}{\log\left( 1-\chooseProbability \chooseProbabilityInv^\nActive \right)} &\leq
\frac{1}{\chooseProbability \chooseProbabilityInv^\nActive} =
(\nActive+1) \left( \frac{\nActive+1}{\nActive} \right)^\nActive \\ &=
(\nActive+1) \left( 1 + \frac{1}{\nActive} \right)^\nActive <
(\nActive+1) \eulerNum,
\end{align*}
where in the last step we have used the fact that $\eulerNum>(1+1/\shortTermVarTwo)^{\shortTermVarTwo}$ for all $x>0$.
\end{proof}
\begin{proof}[Proof of Theorem~\ref{main-theorem}]
By Lemma~\ref{lemma-limit},~(\ref{main-theorem-hypothesis}) implies~(\ref{lemma-expectation-hypothesis}) if we choose $\chooseProbability = (\nActive+1)^{-1}$. So, by Lemma~\ref{lemma-expectation}, we have $\Expectation \potActive{\nSlots} \leq \shareInactive \nActive \errorprob$. Hence, by \Markov's inequality we obtain
\[
\Probability \left( \potActive{\nSlots} \geq \shareInactive \nActive \right) \leq \frac{\Expectation \potActive{\nSlots}}{\shareInactive \nActive} \leq \frac{\shareInactive \nActive \errorprob}{\shareInactive \nActive} = \errorprob. \qedhere
\]
\end{proof}
\section{Simulations}
\label{sec:simulations}
Simulations of the proposed scheme were conducted for three different value pairs of $(\nInactive,\nActive)$. For each value pair and each of $1.2 \cdot 10^5$ simulation runs, the scheme was carried out until only actually active nodes were potentially active and the resulting values of $\nSlots$ were recorded. In Figure~\ref{fig:simulations}, we plotted $\nSlots$ versus the observed frequencies of simulation runs in which the exact set of active nodes was not determined at the receiver in fewer than $\nSlots$ slots. For comparison, we also plotted the corresponding theoretical error bound given by Corollary~\ref{exact-set-corollary}.

For the generation of random numbers called for by the scheme, we used the standard Mersenne Twister pseudo-random number generator shipped with Python~3.5~\cite{Python35PseudoRand} which was seeded before each simulation run with the hardware random number generator provided by Linux~4.4~\cite{Linux44urand}.

From the simulation results, we conclude that the bound provided for $\nSlots$ by Corollary~\ref{exact-set-corollary} is relatively tight for small errors and that it can be achieved using a standard pseudo-random number generator and thus the assumption that receiver and nodes can observe a common source of randomness might be relaxed in practice e.g. by having the receiver broadcast a random seed to all nodes which then use identical and identically seeded random number generators to carry out the scheme.
\begin{figure}
\begin{tikzpicture}
\definecolor{par1color}{HTML}{E66101}
\definecolor{par2color}{HTML}{FDB863}
\definecolor{par3color}{HTML}{5E3C99}
\begin{semilogyaxis} [
  xlabel={Number of slots $\nSlots$},
  ylabel={Error probability bound $\errorprob$ / Error frequency},
  ymin=1e-4,
  ymax=1,
  legend cell align=left,
  legend entries={
    {$(10^4, 20)$},{$(10^5, 20)$},{$(10^4, 30)$}
  }
]
\addplot[par1color,dashed,domain=1:2500,forget plot] {10000/exp(x/exp(1)/21)};
\addplot[par1color,solid,line width=2] table[header=false] {res_N_10000_k_20};

\addplot[par2color,dashed,domain=1:2500,forget plot] {100000/exp(x/exp(1)/21)};
\addplot[par2color,solid,line width=2] table[header=false] {res_N_100000_k_20};

\addplot[par3color,dashed,domain=1:2500,forget plot] {10000/exp(x/exp(1)/31)};
\addplot[par3color,solid,line width=2] table[header=false] {res_N_10000_k_30};
\end{semilogyaxis}
\end{tikzpicture}
\caption{Plots of observed error frequencies in the simulations as thick solid lines and theoretical error bounds according to Corollary~\ref{exact-set-corollary} in thinner dashed lines for various pairs $(\nInactive, \nActive)$.}
\label{fig:simulations}
\end{figure}
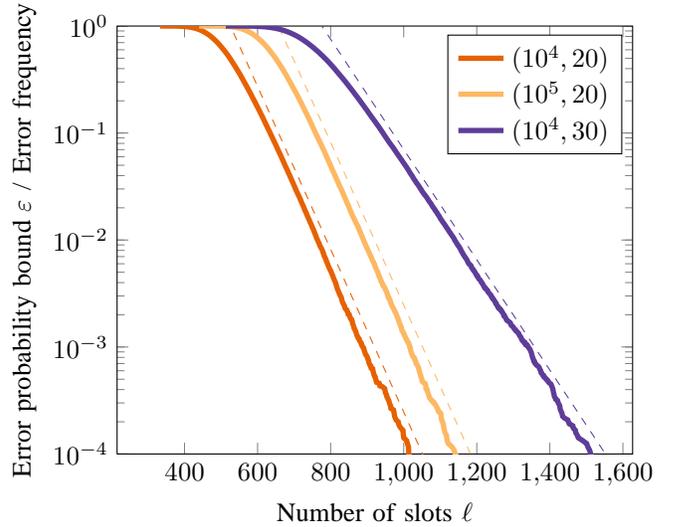
\section{Coding Disjunctions over Channels with Sub-Gaussian Noise}
\label{coding}
In this section, we construct a code for the class of channels with arbitrarily varying sub-gaussian noise (as defined below). This code does not satisfy Assumption~\ref{as:channelcode_existence}, but is close to it as $\channelBlockLength/\nSlots$ grows only logarithmically with $1/\channelErrorprob$. In conjunction with our scheme from Section~\ref{scheme}, it leads to good bounds on the total number of channel uses necessary to solve Problem~\ref{prob:general}.
\begin{definition}
A random variable $\shortTermRV$ is called \emph{sub-gaussian} if for some $\shortTermVarSubgauss \in \reals$, we have
\begin{equation}
\label{eqn:subgauss}
\forall \shortTermVarOne \geq 1~ (\Expectation\absolute{\shortTermRV}^\shortTermVarOne)^{1/\shortTermVarOne} \leq \shortTermVarSubgauss \sqrt{\shortTermVarOne}.
\end{equation}
If $\shortTermRV$ is sub-gaussian, the infimum over all $\shortTermVarSubgauss$ that satisfy~(\ref{eqn:subgauss}) exists and is called the \emph{sub-gaussian norm of $\shortTermRV$} and denoted by $\subgaussnorm{\shortTermRV}$.
\end{definition}
\begin{definition}
A channel as in Definition~\ref{dfn:disjunctioncode} is called a \emph{channel with arbitrarily varying sub-gaussian noise of norm at most $\subgaussnormMax$ and power constraint $\channelPowerLim$} if
\[
\channelInputAlphabet{1} = \dots = \channelInputAlphabet{\nChannelEndpoints} = \left[-\sqrt{\channelPowerLim},\sqrt{\channelPowerLim}\right],~ \channelOutputAlphabet=\reals 
\]
and the transition kernels are such that there exists an independent sequence $(\channelNoise_\channelInstant)_{\channelInstant \in \{1,\dots\}}$ of sub-gaussian random variables (called the \emph{noise}) with mean $0$ and sub-gaussian norm bounded by $\subgaussnormMax$ (in particular, this covers independent and identically distributed Gaussian noise) such that for any time instant $\channelInstant$,
\begin{equation}
\label{eqn:subgausschannel}
\channelOut_\channelInstant = \sum_{\shortTermVarOne=1}^{\nChannelEndpoints} \channelIn{\shortTermVarOne, \channelInstant} + \channelNoise_\channelInstant.
\end{equation}
Transmitters and receiver are not assumed to know anything about the noise random variables except for their independence and the bound $\subgaussnormMax$ on their sub-gaussian norm. Thus the employed coding scheme should work simultaneously and with uniform error bound for any such selection of noise random variables.
\end{definition}
\begin{theorem}
\label{thm:channel-code}
For every $\subgaussnormMax > 0, \channelPowerLim > 0, \channelErrorprob > 0, \nSlots \in \naturals$ and
\begin{equation}
\label{channel-error-probability}
m \geq \frac{\subgaussnormMax^2}{\channelPowerLim} \cdot \frac{\log\frac{1}{\channelErrorprob}+1}{\hoeffdingTypeConst},
\end{equation}
where $\hoeffdingTypeConst>0$ is an absolute constant, there is a code which is an $(\channelrep\nSlots,\nSlots,\channelErrorprob)$-disjunction code for every channel with arbitrarily varying sub-gaussian noise of norm at most $\subgaussnormMax$ and power constraint $\channelPowerLim$.
\end{theorem}
\begin{proof}
For each $\indexChannelEndpoints$, we define a repetition encoder 
\[\channelEncoder{\indexChannelEndpoints}: \channelMessage{\indexChannelEndpoints}=(\channelMessage{\indexChannelEndpoints,1}, \dots, \channelMessage{\indexChannelEndpoints,\nSlots}) \mapsto (\channelIn{\indexChannelEndpoints,\channelInstant})_{\channelInstant \in \{1,\dots,\nSlots\channelrep\}}\]
by insisting that for each slot $\indexSlots \in \{1,\dots,\nSlots\}$
\[
\channelIn{\indexChannelEndpoints,(\indexSlots-1)\channelrep+1} = \dots = \channelIn{\indexChannelEndpoints,\indexSlots\channelrep} =
\begin{cases}
 \sqrt{\channelPowerLim}, &\channelMessage{\indexChannelEndpoints,\indexSlots} = \text{\emph{true}} \\
                       0, &\channelMessage{\indexChannelEndpoints,\indexSlots} = \text{\emph{false}.}
\end{cases}
\]
As an intermediate step, the decoder computes for each $\indexSlots \in \{1,\dots,\nSlots\}$ a value $\channelOutIntermediate_\indexSlots$ as
\[ \channelOutIntermediate_\indexSlots := \frac{1}{\channelrep}\sum_{\channelInstant=(\indexSlots-1)\channelrep+1}^{\indexSlots\channelrep} \channelOut_{\channelInstant}.
\]
We denote the total averaged channel noise in slot $\indexSlots$ by
\[ \tilde{\channelNoise}_\indexSlots = \sum_{\channelInstant=(\indexSlots-1)\channelrep+1}^{\indexSlots\channelrep} \frac{\channelNoise_{\channelInstant}}{\channelrep}
 \]
and note that if the messages $\channelMessage{\indexChannelEndpoints,\indexSlots}$ are \emph{false} for all $\indexChannelEndpoints$ in slot $\indexSlots$, then $\channelOutIntermediate_\indexSlots = \tilde{\channelNoise}_\indexSlots$, whereas if for at least one $\indexChannelEndpoints$ in slot $\indexSlots$ the message $\channelMessage{\indexChannelEndpoints,\indexSlots}$ is \emph{true}, then $\channelOutIntermediate_\indexSlots \geq \sqrt{\channelPowerLim} + \tilde{\channelNoise}_\indexSlots$.
 
So define the decoder $\channelDecoder: (\channelOut_1,\dots,\channelOut_{\nSlots \channelrep}) \mapsto (\channelMessageDecoded_1,\dots,\channelMessageDecoded_\nSlots)$ by
\[
\channelMessageDecoded_\indexSlots :=
\begin{cases}
\text{\emph{true}},  &\channelOutIntermediate_\indexSlots > \frac{1}{2}\sqrt{\channelPowerLim} \\
\text{\emph{false}}, &\text{otherwise.}
\end{cases}
\]
Clearly, $\channelMessageDecoded_\indexSlots$ will be correctly decoded as long as $\absolute{\tilde{\channelNoise}_\indexSlots} < \sqrt{\channelPowerLim}/2$, and the Hoeffding-type inequality for sequences of independent centered sub-gaussian random variables~\cite[Proposition 5.10]{Vershynin_10} tells us that (independently of $\indexSlots$)
\begin{equation}
\label{eqn:hoeffding-type-application}
\Probability\left(\absolute{\tilde{\channelNoise}_\indexSlots} \geq \frac{1}{2}\sqrt{\channelPowerLim}\right) \leq
\eulerNum^{1-\frac{\hoeffdingTypeConst' \channelPowerLim \channelrep}{4\subgaussnormMax^2}},
\end{equation}
where $\hoeffdingTypeConst'>0$ is an absolute constant. By insisting that the right hand side in~(\ref{eqn:hoeffding-type-application}) be less than or equal $\channelErrorprob$, setting $\hoeffdingTypeConst:=\hoeffdingTypeConst'/4$ and solving for $\channelrep$, we obtain~(\ref{channel-error-probability}).
\end{proof}
\begin{cor}
For every $\subgaussnormMax > 0, \channelPowerLim > 0, \errorprob > 0$ and
\begin{align}
\begin{split}
\label{eqn:channel-corollary}
\channelBlockLength &\geq
\frac{\subgaussnormMax^2}{\channelPowerLim} \cdot \frac{1}{\hoeffdingTypeConst} \eulerNum (\nActive+1)  \left( \log\nInactive + \log\frac{1}{\errorprob} \right) \cdot \\ &\phantom{=} \left( 2+ \log(\nActive+1) + \log \log \frac{\nInactive}{\errorprob} + \log\frac{1}{\errorprob} \right),
\end{split}
\end{align}
where $\hoeffdingTypeConst>0$ is an absolute constant, Problem~\ref{prob:general} can be solved over a channel with arbitrarily varying sub-gaussian noise of norm at most $\subgaussnormMax$ and power constraint $\channelPowerLim$ with error bound $2\errorprob$.
\end{cor}
\begin{proof}
We first choose $\nSlots$ to satisfy~(\ref{exact-set-corollary-hypothesis}), and thus, by Corollary~\ref{exact-set-corollary}, if we conduct the scheme with an error free channel code, we can bound the error probability by $\errorprob$. Next, we choose $\channelrep$ according to~(\ref{channel-error-probability}) with $\channelErrorprob=\errorprob/\nSlots$ and use the channel code given by Theorem~\ref{thm:channel-code} to convey our disjunctions. Then, by the union bound, the overall error probability (i.e. the probability of the event that one of our $\nSlots$ disjunctions is decoded incorrectly or the application of the scheme does not result in exact recovery of the set of active nodes) does not exceed $\errorprob + \nSlots \cdot \errorprob / \nSlots = 2 \errorprob$. We can do this as long as we are allowed an overall number of channel uses
\[\channelBlockLength \geq \nSlots \channelrep =
\frac{\subgaussnormMax^2}{\channelPowerLim} \cdot \frac{1}{\hoeffdingTypeConst} \nSlots \left(\log\frac{\nSlots}{\errorprob} + 1\right),\]
and substituting~(\ref{exact-set-corollary-hypothesis}) we get requirement~(\ref{eqn:channel-corollary}).
\end{proof}
\section*{Acknowledgement}
We would like to thank Peter Jung and Richard Küng for a pleasant discussion about their work~\cite{Kueng_Jung_16,Kueng_Jung_16_long} as well as Steffen Limmer for a helpful discussion on possible connections of our work to discrete valued compressive sensing. This work was supported by the German Ministry of Research and Education (BMBF) under grant 16KIS0605 and by the German Research Foundation (DFG) under grant STA 864/8-1.
\bibliographystyle{plain}
\bibliography{references}
\end{document}